\documentclass[12pt]{amsart}
\usepackage{amssymb}
\usepackage{verbatim}
\usepackage[usenames]{color}
\usepackage{hyperref}

\newtheorem{thm}{Theorem}

\newtheorem{corollary}[thm]{Corollary}
\newtheorem{lemma}[thm]{Lemma}
\newtheorem{proposition}[thm]{Proposition}

\theoremstyle{definition}

\theoremstyle{remark}

\renewcommand{\phi}{\varphi}

\newcommand{\ang}[1]{\ensuremath{\langle#1\rangle}}
\newcommand{\ket}[1]{\ensuremath{|#1\rangle}}
\newcommand{\bra}[1]{\ensuremath{\langle#1|}}
\newcommand{\Q}{\smallskip\noindent\textbf{Q: }}
\newcommand{\A}{\smallskip\noindent\textbf{A: }}

\newcommand{\noprint}[1]{\relax}

\newcommand{\R}{\mathbb R}

\title{Negative Probability}
\author{Andreas Blass}
\address{Mathematics Department\\
University of Michigan\\
Ann Arbor, MI 48109--1043, U.S.A.}
\email{ablass@umich.edu}
\author{Yuri Gurevich}
\address{Microsoft Research\\
One Microsoft Way\\
Redmond, WA  98052, U.S.A.}
\email{gurevich@microsoft.com}

\begin{document}
\maketitle

\begin{abstract}
This article was written for the Logic in Computer Science column in the February 2015 issue of the Bulletin of the European Association for Theoretical Computer Science. The intended audience is general computer science audience.

The uncertainty principle asserts a limit to the precision with which position $x$ and momentum $p$ of a particle can be known simultaneously. You may know the probability distributions of $x$ and $p$ individually but the joint distribution makes no physical sense. Yet Wigner exhibited such a joint distribution $f(x,p)$. There was, however, a little trouble with it: some of its values were negative. Nevertheless Wigner's discovery attracted attention and found applications. There are other joint distribution, all with negative values, which produce the correct marginal distributions of $x$ and $p$. But only Wigner's distribution produces the correct marginal distributions for all linear combinations of position and momentum. We offer a simple proof of the uniqueness and discuss related issues.
\end{abstract}

\section{Introduction}
\label{sec:intro}

``Trying to think of negative probabilities,'' wrote Richard Feynman, ``gave me a cultural shock at first'' \cite{Feynman1987}. Yet quantum physicists tolerate negative probabilities. Feynman himself studied, in the cited paper, a probabilistic trial with four outcomes with probabilities $0.6$, $-0.1$, $0.3$ and $0.2$.

We were puzzled. The standard interpretation of probabilities defines the probability of an event as the limit of its relative frequency in a large number of trials. ``The mathematical theory of probability gains practical value and an intuitive meaning in connection with real or conceptual experiments'' \cite[\S I.1]{Feller}. Negative probabilities are obviously inconsistent with the frequentist interpretation. Of course, that interpretation comes with a tacit assumption that every outcome is observable. In quantum physics some outcomes may be unobservable. This weakens the frequentist argument against negative probabilities but does not shed much light on the meaning of negative probabilities.

In the discrete case, a probabilistic trial can be given just by a set of outcomes and a probability function that assigns nonnegative reals to outcomes. One can generalize the notion of probabilistic trial by allowing negative values of the probabilistic function. Feynman draws an analogy between this generalization and the generalization from positive numbers, say of apples, to integers. But a negative number of apples may be naturally interpreted as the number of apples owed to other parties. We don't know any remotely natural interpretation of negative probabilities.

We attempted to have a closer look on what goes on.

Heisenberg's uncertainty principle asserts a limit to the precision with which position $x$ and momentum $p$ of a particle can be known simultaneously: $\displaystyle \sigma_x \sigma_p \ge \hbar/2$ where $\sigma_x, \sigma_p$ are the standard deviations and $\hbar$ is the (reduced) Planck constant. You may know the probability distributions (or the density function or the probability function; we will use the three terms as synonyms) of $x$ and $p$ individually but the joint probability distribution with these marginal distributions of $x$ and $p$ makes no physical sense\footnote{For general information about joint probability distributions and their marginal distributions see \cite[\S IX.1]{Feller}}. Does it make mathematical sense? More exactly, does there exist a joint distribution with the given marginal distributions of $x$ and $p$.

In 1932, Eugene Wigner exhibited such a joint distribution \cite{Wigner1932}. There was, however, a little trouble with Wigner's function. Some of its values were negative. The function, Wigner admits, ``cannot be really interpreted as the simultaneous probability for coordinates and momenta.'' But this, he continues, ``must not hinder the use of it in calculations as an auxiliary function which obeys many relations we would expect from such a probability'' \cite{Wigner1932}.  Probabilistic functions that can take negative values became known as \emph{quasi-probability distributions}.

Richard Feynman described a specific quasi-probability distribution for discrete quantities, two components of a particle's spin \cite{Feynman1987}. The uncertainty principle implies that these two quantities can't have definite values simultaneously. So it seems plausible that an attempt to assign joint probabilities would again, as in Wigner's case, lead to something strange --- like negative probabilities. ``Trying to think of negative probabilities gave me a cultural shock at first \dots It is usual to suppose that, since the probabilities of events must be positive, a theory which gives negative numbers for such quantities must be absurd. I should show here how negative probabilities might be interpreted'' \cite{Feynman1987}. His attitude toward negative probabilities echoes that of Wigner: a quasi-probability distribution may be used to simplify intermediate computations. The meaning of negative probabilities remains unclear. Those intermediate computations may not have any physical sense. But if the final results make physical sense and can be tested then the use of a quasi-probability is justified.

It bothered us that both Wigner and Feynman apparently pull their quasi-probability distributions from thin air. In particular, Wigner writes that his function ``was found by L. Szil\'{a}rd and the present author some years ago for another purpose'' \cite{Wigner1932}, but he doesn't give a reference, and he doesn't give even a hint about what that other purpose was. He also says that there are lots of other functions that would serve as well, but none without negative values. He adds that his function ``seemed the simplest.''

We investigated the matter and made some progress. We found a characterization of Wigner's function that might be considered objective.

\begin{proposition}[Wigner Uniqueness] \label{pro:unique}
Wigner's function is the unique quasi-distribution on the phase space that yields the correct marginal distributions not only for position and momentum but for all their linear combinations.
\end{proposition}

\begin{quotation}
\smallskip\noindent\textbf{Quisani%
\footnote{Readers of this column may remember Quisani, an inquisitive former student of the second author.}: }%
Wait, I don't understand the proposition.
Wigner's function is not a true distribution, so the notion of marginal distribution of Wigner's function isn't defined. Also, what does it mean for a marginal to be correct?

\smallskip\noindent\textbf{Authors%
\footnote{speaking one at a time}: }%
The standard definition of marginals works also for quasi-probability distributions.
A marginal distribution is correct if it coincides with the prediction of quantum mechanics. We'll return to these issues in \S\ref{sec:j2m} and \S\ref{sec:wigner} respectively.

\Q To form a linear combination $ax+bp$ of position $x$ and momentum $p$ you add $x$, which has units of length like centimeters, and $p$, which has momentum units like gram centimeters per second. That makes no sense.

\A We are adding $ax$ and $bp$.  Take $a$ to be a momentum and $b$ to be a length; then both $ax$ and $bp$ have units of action (like gram centimeters squared per second), so they can be added.  If you want to make $ax+bp$ a pure number, divide by $\hbar$.

\Q Finally, is it obvious that Wigner's quasi-distribution is not determined already by the correct marginal distributions for just the position and momentum, without taking into account other linear combinations?

\A This is obvious. There are modifications of Wigner's quasi-distribution that still give the correct marginal distributions for position and momentum. For an easy example, choose a rectangle $R$ centered at the origin in the $(x,p)$ phase plane and modify Wigner's $f(x,p)$ by adding a constant $c$ (resp. subtracting $c$) when $(x,p)\in R$ and the signs of $x$ and $p$ are the same (resp. different).  For a smoother modification, you could add $cxp\exp(-ax^2-bp^2)$ where $a,b,c$ are positive constants (of appropriate dimensions).
\end{quotation}

The idea to consider linear combinations of position and momentum came from Wigner's paper \cite{Wigner1932} where he mentions that projections to such linear combinations preserve the expectations. In fact, the projections give rise to the correct marginals. This led us to the proposition.

In the case of Feynman's quasi-distribution mentioned above, one can't use linear combinations of those two spin components to characterize the distribution. Nor is there a characterization using the spin component in yet another direction. Furthermore, if we only require the correct marginal distributions for the $x$ and $z$ spins, then there are genuine, nonnegative joint probability distributions with those marginals.

Our investigation was supplemented by digging into the literature and talking to our colleagues, especially Nathan Wiebe. That brought us to ``Quantum Mechanics in Phase Spaces: An Overview with Selected Papers'' \cite{Zachos+}. It turned out that there was another, much earlier, approach to characterizing the Wigner quasi-probability distribution. The main ingredient for that earlier approach is a proposal by Hermann Weyl \cite[\S IV.14]{Weyl} for associating Hermitian operators on $L^2$ to well-behaved functions $g(x,p)$ of position and momentum. Jos\'e Enrique Moyal used Weyl's correspondence to characterize Wigner's quasi-distribution in terms of only the expectation values but for a wider class of functions rather than the marginal distributions for just the linear functions of position and momentum \cite{Moyal}.
There is a trade-off here. The class of functions is wider but the feature to match is narrower.

George Baker proved that any quasi-distribution on the position-momentum phase-space, satisfying his ``quasi-probability distributional formulation of quantum mechanics,'' is the Wigner function \cite{Baker}.
The problem of an objective characterization of Wigner's function also attracted the attention of Wigner himself \cite{Wigner1971,OW}.

The volume \cite{Zachos+} does not contain ``our'' characterization of Wigner's function but it is known and due to Jacqueline and Pierre Bertrand \cite{Bertrand}. They found an astute name for the approach: tomographic. The tomographic approach gives an additional confirmation of the fact that behavior of our quantum system is not classical. The approach can be used to establish that the behavior of some other quantum systems is not classical. It has indeed been used that way in quantum optics; see \cite{Smithey,Ourjoumtsev} for example.

Still, in our judgment, our proof of the uniqueness of Wigner's function is simpler and more direct than any other in the literature, and so we present it in \S\ref{sec:wigner}. In section \S\ref{sec:weyl} we establish Moyal's characterization of Wigner's function.

\S\ref{sec:feynman} contains a cursory discussion related to Feynman's four-outcome quasi-distribution. All our observations on that issue happened to be known as well, but we have yet to research the history of the negative probabilities in the discrete case. We intend to address the discrete case elsewhere.

\begin{quotation}

\Q So what is the meaning of negative probability?

\A We don't know.

\Q The use of negative probabilities to validate the quantum character of a quantum system reminds me proofs by contradiction. Assume that the behavior is classical, produce a unique joint distribution, prove the existence of negative values and establish a contradiction. If this is the only use of negative probabilities then there is no need to interpret them semantically.

\A There are some attempts to use negative probabilities as a measure of ``quantumness'' \cite{Veitch1,Veitch2}. We think that the jury is still out.

\Q I have yet another question. Recently, in this very column, Samson Abramsky wrote about contextuality which is another manifestation of non-classical behavior of quantum systems \cite{Abramsky}. I wonder what is the relation, if any, between negative probabilities and contextuality.

\A The discussion of that relation is beyond the scope of this paper. But please have a look at Robert Spekkens's article \cite{Spekkens} with a rather telling title ``Negativity and contextuality are equivalent notions of
nonclassicality.''

\end{quotation}

\subsection*{Acknowledgment}

Many thanks to Nathan Wiebe who was our guide to the literature and the state of art on quasi-probabilities.
\section{Preliminaries}
\label{sec:pre}

We tried to make the paper as accessible as possible; hence this section. We still assume some familiarity with mathematical analysis. By default in this paper integrals are from $-\infty$ to $+\infty$. The baby quantum theory that we use is covered in \S3 of book \cite{Hall} titled ``A first approach to quantum mechanics,''.

\subsection{Fourier transform}
\label{sub:fourier}

The forward Fourier transform sends a function $f(x)$ to
\[
 \hat f(\xi)  =
    \frac1{\sqrt{2\pi}}\int f(x)\,e^{-i\xi x}\,dx.
\]
and the inverse Fourier transform sends a function $g(\xi)$ to
\[
 \check g(x) =
    \frac1{\sqrt{2\pi}} \int g(\xi)e^{i\xi x}\,d\xi,\\
\]
Mathematically $x$ and $\xi$ are real variables. In applications, the dimension of $\xi$ is the inverse of that of $x$ so that $\xi x$ is a pure number.

The forward and inverse Fourier transforms are defined also for functions of several variables. In particular,
\begin{align*}
 \hat f(\xi,\eta)&= \frac1{2\pi}
    \iint f(x,y)\, e^{-i(\xi x + \eta y)}\,dx\,dy,\\
 \check g(x,y)&= \frac1{2\pi} \iint g(\xi,\eta) e^{i(\xi x + \eta y)}\, d\xi\,d\eta.
\end{align*}

\begin{quotation}
\Q What about the convergence of the integrals? Are you going to ignore such details?

\A Yes, we are going to ignore such details. But Fourier transforms are used, with full mathematical rigor, even in some situations where the integrals don't converge.

\Q I do not understand this.

\A The idea is to first define the Fourier transform as an operator on nice functions in $L^2(\R)$, for which the integrals clearly converge. Informally a function $f(x)$ is nice if it and its derivatives $f'(x), f''(x), f'''(x), \dots$ approach zero very rapidly as $x\to\infty$.
The Fourier transform is an isometry on these nice functions, and the nice functions are dense in $L^2(\R)$, so the isometry extends to all of $L^2$. Details can be found in books on real analysis, like \cite{Kingman+} and \cite{Rudin}; alternatively, see \cite[Appendix~A.3.2]{Hall}.

\end{quotation}

\subsection{Dirac's delta function}
\label{sub:delta}

Dirac's $\delta$-function is a generalized function such that for any nice function $f$,
\[
 \int f(x) \delta(x) dx = f(0).
\]
It follows that
\[
 \int f(x) \delta(x-a) dx =
 \int f(x+a) \delta(x) dx = f(a).
\]

Some divergent integrals, e.g. $\int e^{itx}dt$, can be seen as generalized functions in that sense. In fact, as generalized functions,
\[
  \int e^{itx}dt = 2\pi\delta(x).
\]

Indeed,
\begin{align*}
\int dx\,f(x)\,\int e^{itx}dt
  &= \sqrt{2\pi} \int dt\,
     \frac1{\sqrt{2\pi}} \int f(x) e^{itx} dx\\
  &= \sqrt{2\pi} \int \check f(t)\,dt\\
  &= 2\pi \cdot \frac1{\sqrt{2\pi}}
  \int \check f(t)e^{-it0}dt = 2\pi f(0).
\end{align*}

\begin{quotation}
\Q Are these nice functions the same as the nice functions mentioned earlier.

\A Yes, they are.
\end{quotation}

\subsection{Exponential operators}
\label{sub:exponential}

The exponential $e^O$ of an operator $O$ over a topological vector space is the operator
\[
e^O = \sum_{k=0}^\infty \frac{O^k}{k!}
= I + O + \frac12 O^2 + \frac16 O^3 + \dots
\]

If $(X\psi)(x)= x\cdot\psi(x)$
then $(e^X\psi(x))= e^x \psi$, because
\[
(e^X\psi)(x)
= \sum_{k=0}^\infty \frac1{k!} (X^k\psi)(x)
= \psi \cdot \sum_{k=0}^\infty \frac1{k!} x^k
= \psi \cdot e^x.
\]

If $D$ is the derivative operator $\frac d{dx}$, then\\
  $e^{aD}\psi(x) = \psi(x+a)$.
Indeed,
\begin{align*}
 e^{aD} \psi(x)
&=\sum_{k=0}^\infty \frac{(aD)^k \psi(x)}{k!}
 = \sum_{k=0}^\infty \frac{D^kf(x)}{k!} a^k \\
&= \psi(x) + \frac{\psi'(x)}{1!}a
        + \frac{\psi''(x)}{2!}a^2
        + \frac{\psi'''(x)}{3!}a^3 + \dots
\end{align*}
which is the Taylor series of $\psi(x+a)$ around point $x$. (Think of $a$ as $\Delta x$.)

\begin{quotation}

\Q What functions $f$ are you talking about? The Taylor series expansion of $f$ suggests that $f$ is analytic, that is real-analytic.

\A Our intention is that $f$ ranges over $L^2$. By the proof above, $e^{aD}$ is a shift $f(x)\mapsto f(x+a)$ on analytic functions. In particular, $e^{aD}$ is a shift on Gaussian functions
\[
  \exp\left(-\frac{(x-b)^2}{2c^2}\right).
\]
But Gaussian functions span a dense subspace of $L^2(\R)$, and there is a unique continuous extension of $e^{aD}$ to $L^2$, namely the shift $f(x)\mapsto f(x+a)$.

\Q The exponential $e^O$ has got to be a partial operator in general.

\A Yes, $e^O(x)$ is defined whenever the oprators $O^k$ are defined at $x$, and the series
$\sum_{k=0}^\infty \frac{O^k(x)}{k!}$ converges.

\end{quotation}

\section{Joint-to-Marginal Lemma}
\label{sec:j2m}

Let $f(x,p)$ be an ordinary probability distribution or a quasi-distribution on $\R^2$. For any $z = ax+bp$ where $a,b$ are not both zero, the marginal distribution $g(z)$ of $z$ can be defined thus:
\[
g(z) =
\begin{cases}
 \displaystyle
 \frac1b \int f(x, \frac1b (z-ax))\,dx
 &\mbox{if $b\ne0$}\\
 \displaystyle
 \frac1a \int f(\frac1a (z-by), p)\,dp
 &\mbox{otherwise}
\end{cases}
\]
Here's a justification in the case $b\ne0$. We have
\begin{align*}
 p  &= \frac1b(z-ax),\\
 dp &= \frac1b(dz-a\,dx),\\
 f(x,p)\,dx\,dp &= f\big(x,\frac1b(z-ax)\big)\frac1b\,dx\,dz.
\end{align*}
We are relying here on the formalism of differential 2-forms \cite{Flanders} for area elements, so that $dx\,dz$ really means $dx\wedge dz$ and we have used that $dx\wedge dx=0$.  The use of differential forms makes computations like this easier, and it fits well with physics, e.g., with Maxwell's equations and with general relativity.  One could, however, avoid differential forms here and get the same result by considering the Jacobian determinant of the change of variables.

For any real $u\le v$, the probability that $u\le z\le v$ should be
\begin{align*}
\int_u^v g(z) dz
&= \iint_{u\le ax+bp\le v} f(x,p)\,dx\,dp \\
&= \iint_{u\le ax+bp\le v}
   \frac1b f\big(x,\frac1b(z-ax)\big)\,dx\,dz \\
&= \int_u^v dz \int_{-\infty}^\infty
   \frac1b f\big(x,\frac1b(z-ax)\big)\,dx.
\end{align*}
Since the first and last expressions coincide for all $u\le v$, we have
\[
  g(z) = \frac1b \int f\big(x,\frac1b(z-ax)\big)\,dx.
\]

\begin{lemma}[J2M]
For any $a,b$ not both zero, the following statements are equivalent.
\begin{enumerate}
\item $g(z)$ is the marginal distribution of $z = ax + bp$.
\item $\displaystyle \hat g(\zeta) =
  \sqrt{2\pi}\cdot\hat f(a\zeta,b\zeta).$
\end{enumerate}
\end{lemma}

\begin{proof}
To prove (1)$\to$(2), suppose (1) and compare the forward Fourier transforms of $g$ and $f$:
\begin{align*}
 \hat g(\zeta)&= \frac1{\sqrt{2\pi}}\int g(z)e^{-i\zeta z}\,dz\\
 &=\frac1{\sqrt{2\pi}}\iint f\big(x,\frac1b(z-ax)\big)
   e^{-i\zeta z}\frac1b\,dx\,dz\\
 &=\frac1{\sqrt{2\pi}}\iint f(x,p)e^{-i\zeta(ax+bp)}\,dx\,dp.\\[1ex]
 \hat f(\xi,\eta)&=
   \frac1{2\pi}\,\iint f(x,p)e^{-i(\xi x+\eta p)}\,dx\,dp.
\end{align*}
We have
$\hat g(\zeta) = \sqrt{2\pi} \hat f(a\zeta,b\zeta)$.

To prove (2)$\to$(1), suppose (2) and use the implication (1)$\to$(2). If $h$ is the marginal distribution of $z = ax + by$ then
\[
 \hat h(\zeta) = \sqrt{2\pi}\cdot\hat f(a\zeta,b\zeta) =
 \hat g(\zeta),
\]
and therefore $g = h$.
\end{proof}

\begin{corollary}\label{cor:j2m}
For any real $\alpha,\beta$ not both zero,
$\hat f(\alpha,\beta) =
\frac1{\sqrt{2\pi}}\, \hat g(\zeta)$
where
$g(z)$ is the marginal distribution for the linear combination $z=ax+bp$ such that $\alpha=a\zeta$, $\beta=b\zeta$ for some $\zeta$.
\end{corollary}

\section{Wigner uniqueness}
\label{sec:wigner}

The purpose of this section is to prove the Wigner Uniqueness proposition.
For simplicity we work with one particle moving in one dimension, but everything we do in this section generalizes in a routine way to more particles in more dimensions.

In classical mechanics, the position $x$ and momentum $p$ of the particle determine its current state. The set of all possible states is the phase space of the particle. By Corollary~\ref{cor:j2m}, an ordinary distribution $f(x,p)$ on the phase space is uniquely determined by its marginal distributions for all linear combinations $ax+bp$ where $a,b$ are not both zero.

In the quantum case, a state of the particle is given by a normalized (to norm 1) vector $\ket\psi$  in $L^2(\R)$. The position and momentum are given by Hermitian operators $X$ and $P$ where
\[
 (X\psi)(x)= x\cdot\psi(x)\quad\text{and}\quad
 (P\psi)(x)= -i\hbar\frac{d\psi}{dx}(x).
\]
For any $a,b$ not both zero, the linear combination $z = ax + by$ is given by the Hermitian operator $Z=aX+bP$. In a state $\ket\psi$, there is a probability distribution $g(z)$ (for the measurement) of the values of $z$. For a function $h(z)$ of $z$, the expectation of $h(z)$ is $\ang{\psi | h(Z) | \psi}$.

The following technical lemma plays a key role in our proof of the uniqueness of Wigner's quasi-distribution.

\begin{lemma}\label{lem:key}
\[
 \bra\psi e^{-i (\alpha X + \beta P)}\ket\psi =
 e^{i\alpha\beta\hbar/2}
 \int\psi^*(y)e^{-i\alpha y}\psi(y-\beta\hbar)\,dy.
\]
\end{lemma}

\begin{proof}
We want to split the exponential into a factor with $X$ times a factor with $P$. This is not as easy as it might seem, because $X$ and $P$ don't commute.  We have, however, two pieces of good luck.  First, there is Zassenhaus's formula, which expresses the exponential of a sum of non-commuting quantities as a product of (infinitely) many exponentials, beginning with the two that one would expect from the commutative case, and continuing with exponentials of nested commutators:
\[
 e^{A+B}=e^Ae^Be^{-\frac12[A,B]}\cdots,
\]
where the ``$\cdots$'' refers to factors involving double and higher commutators.

\begin{quotation}
\Q You gave no reference to Zassenhaus's paper.

\A Apparently, Zassenhaus never published this result, but there's a paper \cite{Casas+} that shows how to compute the next terms.  It also has a pointer to early uses of the formula.
\end{quotation}

The second piece of good luck is that $[X,P]=i\hbar I$, where $I$ is the identity operator.  (In the future, we'll usually omit writing $I$ explicitly, so we'll regard this commutator as the scalar $i\hbar$.) Since that commutes with everything, all the higher commutators in Zassenhaus's formula vanish, so we can omit the ``$\cdots$'' from the formula. We have
\[
 \bra\psi e^{-i\alpha X-i\beta P}\ket\psi\\
   = \bra\psi e^{-i\alpha X}e^{-i\beta P}e^{\alpha\beta[X,P]/2}\ket\psi.
\]
The last of the three exponential factors here arose from Zassenhaus's formula as
\[
 -\frac12[-i\alpha X,-i\beta P] = \frac12\alpha\beta[X,P] = i\alpha\beta\hbar/2.
\]
That factor, being a scalar, can be pulled out of the bra-ket. Taking into account \S\ref{sub:exponential},
\[
 \bra\psi e^{-i (\alpha X + \beta P)}\ket\psi=
 e^{i\alpha\beta\hbar/2}
 \int \psi^*(y) e^{-i\alpha y} \psi(y-\beta\hbar)\,dy.
\]
\end{proof}

Now we are ready to prove the Wigner Uniqueness proposition.
Suppose that a quasi-distribution $f(x,p)$ yields correct marginal distributions for all linear combinations of position and momentum. For any real $\alpha,\beta$ not both zero, let $a,b,g,\zeta$ be as in Corollary~\ref{cor:j2m}. Then
\begin{equation}\label{1}
\begin{aligned}
\hat f(\alpha,\beta)
= \frac1{\sqrt{2\pi}} \hat g(\zeta)
&=\frac1{2\pi}\int g(z)e^{-i\zeta z}\,dz\\
 &=\frac1{2\pi} \ang{e^{-i\zeta Z}}
  =\frac1{2\pi}\bra\psi e^{-i\zeta Z}\ket\psi\\
 &=\frac1{2\pi}\bra\psi e^{-i\zeta (aX+bP)}\ket\psi.
\end{aligned}
\end{equation}
By Lemma~\ref{lem:key},
\begin{equation}\label{2}
\hat f(\alpha,\beta)
= \frac{e^{i\alpha\beta\hbar/2}} {2\pi}
   \int\psi^*(y)e^{-i\alpha y}\psi(y-\beta\hbar)\,dy.
\end{equation}
To get $f(x,p)$, apply the (two-dimensional) inverse Fourier transform.
\[
 f(x,p)=\frac1{(2\pi)^2}\iiint
 \psi^*(y)e^{-i\alpha y}e^{i\alpha\beta\hbar/2}\psi(y-\beta\hbar)
 e^{i\alpha x}e^{i\beta p}\,dy\,d\alpha\,d\beta.
\]
Collecting the three exponentials that have $\alpha$ in the exponent,
and noting that $\alpha$ appears nowhere else in the integrand,
perform the integration over $\alpha$ and (recall \S\ref{sub:delta}) get a Dirac delta function:
\[
 \int e^{-i\alpha(y-\frac{\beta\hbar}2-x)}\, d\alpha=2\pi\delta(y-x-\frac{\beta\hbar}2).
\]
That makes the integration over $y$ trivial, and what remains is
\begin{equation}\label{wigner}
f(x,p)=\frac1{2\pi}\int\psi^*(x+\frac{\beta\hbar}2)\psi(x-\frac{\beta\hbar}2)
e^{i\beta p}\,d\beta,
\end{equation}
which is Wigner's quasi-distribution.

To check that Wigner's quasi-distribution yields correct marginal distribution note that the derivation of \eqref{wigner} from \eqref{1} is reversible.
This completes the proof of the Wigner Uniqueness proposition.

\section{Weyl's correspondence}
\label{sec:weyl}

There is another approach to characterizing the Wigner quasi-probability distribution, using the expectation values for a wide class of functions rather than the marginal distributions for just the linear functions of position and momentum.  The main ingredient for this approach is a proposal by Hermann Weyl \cite[\S IV.14]{Weyl} for associating a Hermitian operator on $L^2(\R)$ to any (well-behaved) function $g(x,p)$ of position and momentum.  Weyl's proposal is to first form the Fourier transform $\hat g(\alpha,\beta)$ of $g(x,p)$, and then apply the inverse Fourier transform with the Hermitian operators $X$ and $P$ in place of the classical variables $x$ and $p$.  Thus, the Weyl correspondence associates to $g(x,p)$ the operator
\[
g(X,P) = \frac1{2\pi}
         \iint \hat g(\alpha,\beta)
e^{i(\alpha X+\beta P)}\,d\alpha\,d\beta.
\]
If one grants that this is a reasonable way of converting phase-space functions $g(x,p)$ to operators $g(X,P)$, then a desirable property of a phase-space quasi-probability distribution $f(x,p)$ would be that the expectation of $g(X,P)$ in a quantum state \ket\psi\ is the same as the expectation of $g(x,p)$ under $f(x,p)$.
We shall show that the Wigner distribution is uniquely characterized by enjoying this desirable property for all well-behaved $g$.

Indeed, the expectation of $g(X,P)$ in state \ket\psi\ is
\[
\bra\psi g(X,P)\ket\psi =
\frac1{2\pi} \iint\hat g(\alpha,\beta)
\bra\psi e^{i(\alpha X+\beta P)}\ket\psi\,d\alpha\,d\beta,
\]
and the expectation of $g(x,p)$ under the distribution $f(x,p)$ is
\[
\iint g(x,p)f(x,p)\,dx\,dp =
\iint\hat g(\alpha,\beta)\hat
f(\alpha,\beta) \,d\alpha\,d\beta.
\]
This last equation is a consequence of the fact, mentioned in \S\ref{sub:fourier}, that the Fourier transform is a unitary operator and therefore preserves the inner product structure of $L^2(\R)$.  Since these two expectations agree for all (well-behaved) $g$,
\[
\hat f(\alpha,\beta) = \frac1{2\pi}
\bra\psi e^{i(\alpha X+\beta P)}\ket\psi.
\]
But this is the part of equation\eqref{1} that was used to derive Wigner's formula \eqref{wigner}.

\begin{quotation}
\Q I wonder how Weyl arrived at his proposal.

\A Weyl presents his proposal in \cite{Weyl} without any motivation, so we don't know how he came up with it, but we can speculate.  The title of \cite{Weyl} indicates that Weyl was working in a group-theoretic context.  As a result, the Fourier transform, expressing functions on $\R$ as combinations of the characters $e^{i\alpha x}$ of the group $(\R,+)$ would be in the forefront of his considerations.  Now consider his goal --- to somehow convert a classical function $g(x,p)$ into an operator.  Roughly speaking, he would want to substitute the operators $X$ and $P$ for the classical variables $x$ and $p$. An obvious difficulty is that the same functions $g(x,p)$ might have two different expressions, for example $xp=px$, which are no longer equivalent when operators are substituted, $XP\neq PX$.  So it is reasonable to try to choose, from the many expressions for a function $g(x,p)$, one particular, reasonably canonical expression, into which one can substitute $X$ and $P$.  The Fourier expansion, $\int \hat g(\alpha,\beta)\exp(i(\alpha x+\beta p))\,dx\,dp$ has those properties.  It depends only on the function $g$, not on how one chooses to express it, and there is no problem substituting $X$ and $P$ for $x$ and $p$.
\end{quotation}

\section{Feynman and spins}
\label{sec:feynman}

Richard Feynman studied ``an analogue of the Wigner function for a spin $\frac12$ system or other two state system'' \cite{Feynman1987}.  He chose the $z$ and $x$ components of the spin to serve as the analogs of the position and momentum in Wigner's formula.

\Q His case should be much simpler than Wigner's case.

\A Not necessarily. While the commutator $[X,P]$ is a scalar, the commutator of the $z$ and $x$ components of a spin is the $y$-component times $i\hbar$. This is but one of several complications.

\Q Is Feynman's quasi-distribution determined by the correct marginals for all linear combinations of the $x$ and $z$ spins?

\A That question sounds reasonable until you look at it a little more closely.
To fix notation, let's describe spin by means of the standard Pauli matrices
\[
X=
\begin{pmatrix}
  0&1\\1&0
\end{pmatrix},\qquad Y=
\begin{pmatrix}
  0&-i\\i&0
\end{pmatrix},\qquad Z=
\begin{pmatrix}
  1&0\\0&-1
\end{pmatrix}.
\]
The usual matrix representation of the spins for a spin $\frac12$ particle is given by these matrices divided by 2, but it's convenient to skip those extra factors $\frac12$; if you like, imagine that we are measuring angular momentum in units of $\hbar/2$ instead of $\hbar$.

So each of our matrices has eigenvalues $\pm1$, with $+1$ meaning spin along the corresponding positive axis and $-1$ along the corresponding negative axis.  For example, the two basis states of spin up and spin down along the $z$ axis are the eigenvectors for eigenvalues $1$ and $-1$ of $Z$; equivalently, they correspond to eigenvalues $1$ and $0$ for $(I+Z)/2$, where $I$ is the identity operator. This point of view is useful because it implies that, in any state $\ket{\psi}$, $(1+\ang Z)/2$ is the probability that the $z$ spin is up. Here, as before, angle brackets denote expectations.  Similarly, $(1-\ang Z)/2$ is the probability that the $z$ spin is down.  Of course, analogous formulas apply to the $x$ and $y$ components of the spin.

Feynman, in analogy to Wigner, introduces a quasi-probability distribution $f$ for the pair of non-commuting observables $Z$ and $X$.  So $f$ has four components, $f_{++},f_{+-},f_{-+}$, and $f_{--}$, as the quasi-probability of $Z$ and $X$ having the values $\pm1$ given by the subscripts of $f$.

Now let's look at a linear combination of $Z$ and  $X$, for example the simplest nontrivial one, $Z+X$.  From the point of view of quasi-probabilities,
\begin{itemize}
  \item with probability $f_{++}$, $Z$ has value 1 and $X$ has value
    1, so $Z+X$ has value 2,
\item with probability $f_{+-}$, $Z$ has value 1 and $X$ has value
    $-1$, so $Z+X$ has value 0,
\item with probability $f_{-+}$, $Z$ has value $-1$ and $X$ has value
    1, so $Z+X$ has value 0, and
\item with probability $f_{--}$, $Z$ has value $-1$ and $X$ has value
    $-1$, so $Z+X$ has value $-2$.
\end{itemize}
Altogether, the possible values of $Z+X$ are 2, 0, and $-2$, with quasi-probabilities $f_{++}$, $f_{+-}+f_{-+}$, and $f_{--}$, respectively.

So your proposed analog of the result for Wigner's distribution would assume that $f$ is chosen so that these probabilities agree, in a given state, with the probabilities computed by quantum mechanics. But such agreement is impossible, because, according to quantum mechanics, the possible values of $Z+X$ are the eigenvalues of this operator, namely $\pm\sqrt2$, which are completely different from the $2,0,-2$ arising from the quasi-probabilities.
It is easy to check that the possible values of any nontrivial linear combination of $Z$ and $X$ are completely different from the values arising from the quasi-probabilities.

\Q OK, let's require the minimum that Feynman obviously intended, namely that $f$ should produce the correct marginal distributions of $Z$ and $X$. Does this determine $f$ uniquely?

\A At first sight, this looks promising. Requiring the correct marginals for the two variables, each having two possible values, gives us four equations for the four unknown components $f_{\pm\pm}$ of $f$:
\begin{align*}
f_{++}+f_{+-}&=\frac12(1+\ang Z)\\
f_{-+}+f_{--}&=\frac12(1-\ang Z)\\
f_{++}+f_{-+}&=\frac12(1+\ang X)\\
f_{+-}+f_{--}&=\frac12(1-\ang X).
\end{align*}
But there's redundancy in the equations; only three of them are independent, so there's one free parameter in the general solution.  In fact, it's easy to write down the general solution:
\begin{align*}
f_{++}&=\frac14(1+\ang Z+\ang X+t)\\
f_{+-}&=\frac14(1+\ang Z-\ang X-t)\\
f_{-+}&=\frac14(1-\ang Z+\ang X-t)\\
f_{--}&=\frac14(1-\ang Z-\ang X+t),
\end{align*}
where $t$ is arbitrary.
Feynman's formulas correspond to $t=\ang Y$ but we see no reason to prefer \ang Y over, for example, $-\ang Y$.

\Q Put the freedom in choosing $t$ to some use.  How about minimizing the negativity in $f$? In other words, adjust $t$ to bring $f$ as close as possible to being a genuine probability distribution.

\A That idea works better than we originally expected. One can get rid of the negativity altogether.  For each state \ket\psi, there is a choice of $t$ that makes all four components of $f$ nonnegative.

Indeed, write down the four inequalities $f_{\pm\pm}\geq0$ using the formulas above for these $f_{\pm\pm}$'s.  Solve each one for $t$. You find two lower bounds on $t$, namely
\begin{align*}
-1-\ang Z-\ang X &\text{ (from $f_{++}\ge0$)}\\
-1+\ang Z+\ang X &\text{ (from  $f_{--}\ge0$)},
\end{align*}
and two upper bounds, namely
\begin{align*}
1+\ang Z-\ang X &\text{ (from $f_{+-}\ge0$)}\\
1-\ang Z+\ang X &\text{ (from $f_{-+}\ge0$)}.
\end{align*}
An appropriate $t$ exists if and only if both of the lower bounds are less than or equal to both of the upper bounds. That gives four inequalities, which simplify to $-1\leq\ang Z\leq 1$ and $-1\leq\ang X\leq 1$.  But these are always satisfied, because the eigenvalues of $Z$ and $X$ are $\pm1$.

\Q I am confused. The uncertainty principle asserts that you cannot measure $Z$ and $X$ at once. Accordingly one would expect that the joint probability distribution $f_{\pm\pm}$ should not exist or, as in Wigner's case, should have at least one negative value.

\A The relevant difference between Wigner's and Feynman's cases seems to be this. In Wigner's case, there is naturally a rich set of marginals that the joint probability distribution is supposed to produce, namely the probability distributions of all linear combinations of the position $x$ and the momentum $p$ of the particle. In Feynman's case, the natural set of marginals is too poor, just the probability distributions of $Z$ and $X$.

\Q Did Feynman find a good use for quasi-probabilities?

\A He introduced negative probabilities in connection to a problem of infinities in quantum field theory. ``Unfortunately I never did find out how to use the freedom of allowing probabilities to be negative to solve the original problem of infinities in quantum field theory!'' \cite{Feynman1987}.

\Q Still, the idea to use quasi-probability distributions to simplify intermediate computations looks attractive to me.

\A You are in good company.

\end{document}